\documentclass[onecolumn,draftclsnofoot,12pt]{IEEEtran}
\usepackage[font=small]{caption}
\usepackage[normalem]{ulem}
\usepackage{enumitem}
\usepackage{xcolor}
\ifCLASSINFOpdf
\usepackage[pdftex]{graphicx}
\usepackage{bm}
\DeclareGraphicsExtensions{.pdf}
\else
\fi
\usepackage{mathtools}
\usepackage{graphicx}
\usepackage{amsthm}
\usepackage{cite}
\usepackage{color}
\usepackage{bbm}
\usepackage{amsmath,amssymb,amsfonts}
\usepackage{subcaption}

\newcommand{\ie}{\textit {i.e., }}

\newcommand{\x}{\pmb{A}}
\newcommand{\y}{\pmb{B}}

\newcommand{\dd}{\mathrm{d}}

\newcommand{\hp}[1]{\mathsf{p}_{#1}}

\newcommand{\hpdiff}[1]{\mathsf{q}_{#1}}
\newcommand{\hpone}{\overline{\mathsf{p}}}
\newcommand{\lhp}[1]{\mathcal{P}_{#1}}
\newcommand{\lhpone}{\overline{\mathcal{P}}}

\newcommand{\erfc}{\mathrm{erfc}}

\newtheorem{theorem}{Theorem}

\newtheorem{coro}{Corollary}[theorem]

\newcommand{\laplace}[2]{\mathcal{L}_{#1}\left[#2\right]}

\usepackage{graphics}
\IEEEoverridecommandlockouts
\begin{document}
	\title{Impact of Multiple Fully-Absorbing Receivers in Molecular Communications}
		\author{Nithin V. Sabu, Abhishek K. Gupta, Neeraj Varshney, Anshuman Jindal
		\thanks{ N. V. Sabu, A. K. Gupta and A. Jindal are with IIT Kanpur, 
 208016, India (Email:{ \{nithinvs,gkrabhi,anshuji\}@iitk.ac.in}). N. Varshney is with NIST, 
Gaithersburg, MD 20899, USA (Email: {neerajv@ieee.org}). 
			This research was supported by the DST-SERB 
(India) under the grant SRG/2019/001459 and IITK under the grant IITK/2017/157. 
		A part of the work 
		 is submitted  in \cite{sabu2021channel}.}}
	\maketitle

\begin{abstract}
Molecular communication is a promising solution to enable intra-body communications among nanomachines. However, malicious and non-cooperative receivers can degrade the performance, compromising these systems' security. Analyzing the communication and security performance of these systems requires accurate channel models. However, such models are not present in the literature. In this work, we develop an analytical framework to derive the hitting probability of a molecule on a fully absorbing receiver (FAR) in the presence of other FARs, which can be either be cooperative or malicious. We first present an approximate hitting probability expression for the 3-FARs case. A simplified expression is obtained for the case when FARs are symmetrically positioned. Using the derived expressions, we study the impact of malicious receivers on the intended receiver and discuss how to minimize this impact to obtain a secure communication channel. We also study the gain that can be obtained by the cooperation of these FARs. We then present an approach to extend the analysis for a system with $ N $ FARs. The derived expressions can be used to analyze and design multiple input/output and secure molecular communication systems.

	\end{abstract}


\section{Introduction}
The future sixth generation of communication standards is envisioned to support nanomachine communication \cite{tripathi2021}. For intra-body communication, molecular communication (MC) is a promising solution for reliable communication among bio-nanomachines (nanomachines with biological functional parts)  \cite{nakano2013}. MC uses molecules as the carrier of information, termed information molecules (IMs). The out of the body communication can be effectively performed with the help of electromagnetic (EM) waves like tera-hertz waves\cite{tripathi2021}. A hybrid system utilizing both MC and EM-based communication can help realize efficient healthcare systems. Here, MC systems can gather information from inside the human body, and  EM-based systems convey the gathered information to a healthcare facility to process the information and take necessary actions. 

In intra-body communication systems, multiple receiver bio-nanomachines may coexist in the same communication channel. For example, multiple malicious receivers may coexist with MC systems inside a human body. The malicious receivers can act as eavesdroppers or hinder the triggering of MC detection systems by absorbing the IMs \cite{eved}. Such actions have profound security implications in e-health systems. Multiple receivers can also be utilized to create single-input multiple-output (SIMO) and multiple-input multiple-output (MIMO) communication links with improved reliability and performance. 

An important metric in characterizing a molecular channel is hitting probability which denotes the probability that an IM hits the intended receiver (or equivalently the fraction of IMs hitting the receiver) within time $t$.
For a diffusion-based MC system with a point transmitter and a fully-absorbing receiver (FAR) in three-dimension (3-D), the hitting probability was discussed in \cite{yilmaz2014}. 
When there are other FARs present in the same medium, they can capture the IMs transmitted to the intended receiver. 
Therefore, the hitting probability decreases with the number of additional receivers, thus degrading the communication reliability. It is important to derive the hitting probability for general $N$ FARs to study the extent of the impact of additional receivers.
However, there is no analytical channel model for an MC system with multiple FARs. Many of the works \cite{koo2016, damrath2018,lee2017} relied on fitting models, machine learning-based approaches, and particle-based simulations to analyze systems with multiple FARs. For a system with two FARs, an approximate analytical expression for the hitting probability on \textit{each of} the FAR was derived in our past work \cite{sabu2020a}. Later this work was extended to derive the approximate expression for hitting probability of an IM on FARs with different sizes \cite{sabu2021ab}. An approximate analytical equation for the hitting probability of an IM on \textit{any one} of the FAR in a system with multiple FARs was derived in \cite{sabu2020}. However, analysis for an arbitrary value of $N$ is still missing in the current literature. 
 Analyzing MC systems with multiple FARs using the methods employed in  \cite{koo2016, damrath2018,lee2017} are time-consuming and inconvenient. Even though obtaining exact analytical expressions for the hitting probability of MC systems with multiple FARs can be difficult, approximate analytical expressions with good accuracy are of a greater value, which is the focus of the work.

This work considers a diffusion-based MC system in 3-D with a point transmitter and $N$ FARs. We first consider $N=3$ and derive  the hitting probability of an IM on one FAR in presence of the rest of the FARs, which is validated using simulations. We then provide the approach for the general $N$ case. 
The contributions of this work are listed below.
\begin{itemize}[leftmargin=10pt]
	\item We consider an MC system with three FARs located in arbitrary positions in 3-D. We derive an approximate yet accurate analytical expression for the hitting probability of an IM on each of the FARs while considering the impact of the rest of the FARs. The derived equation is then validated using particle-based simulations.
	\item Using the derived equation for the 3 FAR system, a very simple expression is obtained when the FARs are symmetrically positioned (\ie mutually equidistant and equidistant from the transmitter).
	\item We also study the impact of malicious receivers on the intended receiver's performance by  quantifying the mutual influence of FARs and discuss when the impact is minimal.
	\item We also study the gain in detection performance when FARs can cooperate with each other.
	\item We also provide a framework to derive the  hitting probability expression for an MC system with  $ N $  FARs for an arbitrary value of $N$.
	
\end{itemize}
\textbf{Notation}: 
 $ \mathcal{L}\left[.\right] $ denotes the Laplace transform and 
 $ \mathcal{L}^{-1}\left[.\right] $ denotes the inverse Laplace transform.

\section{System Model} 
	\begin{figure}[ht]
	\centering   
	\includegraphics[width=0.65\linewidth]{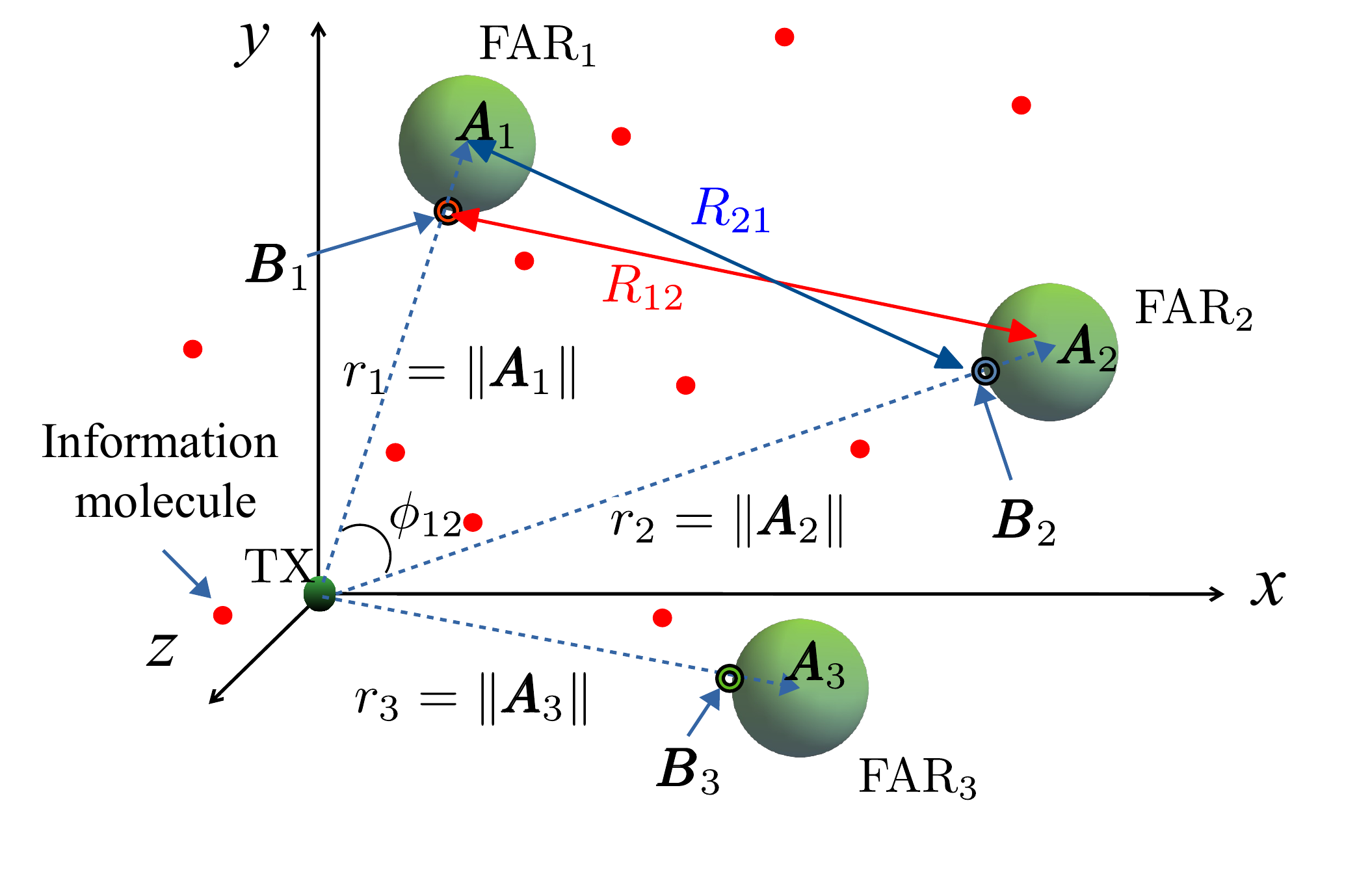}
	\caption{An illustration showing a 3-D MC  system with a point transmitter at the origin and three FARs located at positions $\x_{1},\ \x_{2},\ \text{and},\ \x_{3}$ in $ \mathbb{R}^3 $. }
	\label{Fig:sm}
\end{figure}
We consider a diffusion-limited MC system with a point transmitter located at origin $\mathbf{o}$ and $N=3$ FARs located at $\x_i$'s ($ i\in\{1,2,3\} $) in a 3-D medium as shown in Fig. \ref{Fig:sm}. The radius of all the FARs is assumed to be the same and is denoted by $ a $. 
The IM emitted by the point transmitter propagates in the medium via Brownian motion\cite{einstein2011}. The propagation medium is assumed to be homogeneous, and the diffusion coefficient (denoted by $ D $) of the IM is assumed to be constant over space and time. An IM is absorbed and is detected with a probability $ 1 $ when it hits the surface of FAR. 

The considered setup can be applied to various scenarios, including cooperative and non-cooperative settings. For example, the multiple FARs can be cooperating where they work together as a SIMO/MIMO system to decode the same message transmitted from $\mathbf{o}$. In a different scenario, these FARs can act as malicious or non-cooperative receivers interfering with the other receivers' communication. Without loss of generality, we assume that the first FAR FAR$_1$ is the intended FAR and others are cooperating or non-cooperating FARs. The probability that an IM hits the FAR$_1$ within time $t$  is known as the hitting probability.

Let the distance between the transmitter and the FAR at $ \x_i $ be denoted by $ r_i $, where $ r_i=\|\x_i\|,\ i\in\mathbb{N}$. The angle between $ i $th FAR located at $ \x_i $ (denoted by FAR$ _i $) and FAR$_j $ is represented by $ \phi_{ij},\ i,j\in\mathbb{N},\ i\neq j $. Let us denote the closest point on the surface of FAR$ _i $ from the transmitter by $ \y_{i} $. Therefore the distance between the closest point of FAR$ _i $ from the transmitter  and the center of FAR$ _j $ is
$ 	R_{ij}=\sqrt{(r_i-a)^2+r_{j}^2-2(r_i-a)r_{j}\cos\phi_{ij}}. $

Note that the hitting probability expression is known  for $N=1$ and $N=2$, which are given respectively as  \cite{yilmaz2014,sabu2020a}
	\begin{align}
		\hpone(t,r_1)&=\frac{a}{r_1}\erfc\left(\frac{r_1-a}{\sqrt{4Dt}}\right),  \text{ and}\label{eq1rx}\\
		\hp{}(t,\x_1{\mid} \x_2)&=\sum_{n=0}^{\infty} \frac{a^{2 n}}{R_{12}^{n} R_{21}^{n}}\left[\frac{a}{r_1} \erfc\left(\frac{r_1{-}a+n\left(R_{21}{-}a\right)+n\left(R_{12}{-}a\right)}{\sqrt{4 D t}}\right)\right. \nonumber\\
		&\left.\quad -\frac{a^{2}}{r_2 R_{21}} \erfc\left(\frac{r_2{-}a+(n+1)\left(R_{21}{-}a\right)+n\left(R_{12}{-}a\right)}{\sqrt{4 D t}}\right)\right].
		\label{eq2rx}
	\end{align}
However,  the expression for $N\ge3$ is not available. We will first consider  three FARs (\ie $N=3$) located at $ \x_1,\x_2 $ and $ \x_3 $ and later extend the analysis for the general case.

\section{Channel Model for an MC System with 3-FARs}
In this section, we will analyze an MC system with 3 FARs in terms of hitting probability.
\subsection{Hitting probability}
We now derive the hitting probability of an IM on the FAR in a system with three FARs, located at $\x_1$, $\x_2$ and $\x_3$. It can be seen from \eqref{eq2rx} that $ \hp{}(t,\x_1{\mid} \x_2)\leq \hpone(t,r_1)$, which shows that the hitting probability decreases when an additional receiver is present.  
It is expected that for an MC system with three FARs, the hitting probability would be less than   $ \hp{1}(t,\x_1{\mid} \x_2)$. 
Let us denote the hitting probability until time $ t $ of an IM emitted by a point transmitter located at the origin on FAR$ _1 $ (located at $ \x_1 $) in the presence of the other two FARs by $ \hp{}\left(t,\x_1{\mid} \x_2,\x_3\right)$, 
or just $ \hp{1}(t) $ (when there is no ambiguity) 
which is given in the theorem below.

\begin{theorem}\label{t3rx}
	In an MC system with a point transmitter at the origin and three FARs located at $\x_1,\ \x_2$ and $\x_3$, the probability that an IM emitted from the point transmitter hits  $\mathrm{FAR} _1 $ within time $ t $ in the presence of other two FARs is approximately given by
	\begin{align}
	\hp{1}(t)&=\mathcal{L}^{-1}\left[\lhp{1}(s)\right],\label{eq3rx}
\end{align}
{ where } $\lhp{1}(s)$ is the Laplace transform of $\hp{1}(t)$ and is given by
\begin{align}
	\lhp{1}(s)= \frac{\lhpone(s,r_1) -s\alpha_1(s)+s^2\beta_1(s)}
	{1-s^2\gamma(s)+s^3\delta(s)}\label{lp3rx}. 
\end{align}
 Here,
 \begin{align*}
 \alpha_1(s)&=\lhpone(s,r_2)\lhpone(s,R_{21}){+}\lhpone(s,r_{13})\lhpone(s,R_{31})\nonumber\\
 	\beta_1(s)&=-\lhpone(s,r_1)\lhpone(s,R_{23})\lhpone(s,R_{32})+\lhpone(s,r_2)\lhpone(s,R_{23})\lhpone(s,R_{31})\nonumber\\
 	 &\qquad+\lhpone(s,r_{3})\lhpone(s,R_{32})\lhpone(s,R_{21})\nonumber\\
 	\gamma(s)&=\lhpone(s,R_{12})\lhpone(s,R_{21}) + \lhpone(s,R_{32})\lhpone(s,R_{23})
	+\lhpone(s,R_{13}) \lhpone(s,R_{31})\nonumber\\
 	\text{and  }
 	\delta(s)&=\lhpone(s,R_{12}) \lhpone(s,R_{23})\lhpone(s,R_{31})
	+ \lhpone(s,R_{13})\lhpone(s,R_{32})\lhpone(s,R_{21})\nonumber.
 \end{align*} 
Here, $\lhpone(s,x)
$ is the Laplace transforms of  $\hpone(t,x)$ 
given as 
\begin{align*}
	\lhpone(s,x)&=\laplace{}{\hpone(t,x)}=\frac{a}{sx}\exp\left(-\left(x-a\right)\sqrt{\frac{s}{D}}\right)\nonumber\hspace{3.5in}\ \quad
\end{align*}
with $\ x\in\{r_j,R_{ij}\}$ for all $i,j\in\{1,2,3\}, i\neq j
$.
\begin{proof}
	See Appendix \ref{a3rx}.
\end{proof}  
\end{theorem}

The source of approximation in the above theorem is the approximation of the hitting point of an IM on the surface of the FAR with the closest point on FAR as given in Appendix \ref{a3rx}. This approximation is of good accuracy and is verified in Section \ref{acc}.\par

\subsection{Evaluation of the hitting probability for $ 3-$FAR case}

Fig. \ref{fig:3rxf1} shows the variation of the hitting probability within time $ t $ with time $ t $ for an MC system with three FARs. The validity of the Theorem \ref{t3rx} is verified using particle-based simulation in MATLAB. In the particle-based simulation for the considered system, an IM is generated at the transmitter at the origin at time $ t=0 $. The IM undergoes Brownian motion with diffusion coefficient $ D $. The position of the IM is tracked for every time step $ \Delta t $. If the distance between the IM and the center of the FAR is less than the radius of the FAR, the IM is considered to be hit at the surface of the FAR. The hitting process is accurate when $ \Delta t $ is too small. This procedure is continued for many iterations, and the mean of the hitting events is taken to get the hitting probability. 
 Fig. \ref{fig:3rxf1} verifies that the derived equations are in good match with the simulation results. We can also observe that the hitting probability of FARs increases when it is near the transmitter.

\begin{figure}[ht]
	\centering
	\includegraphics[width=0.65\linewidth]{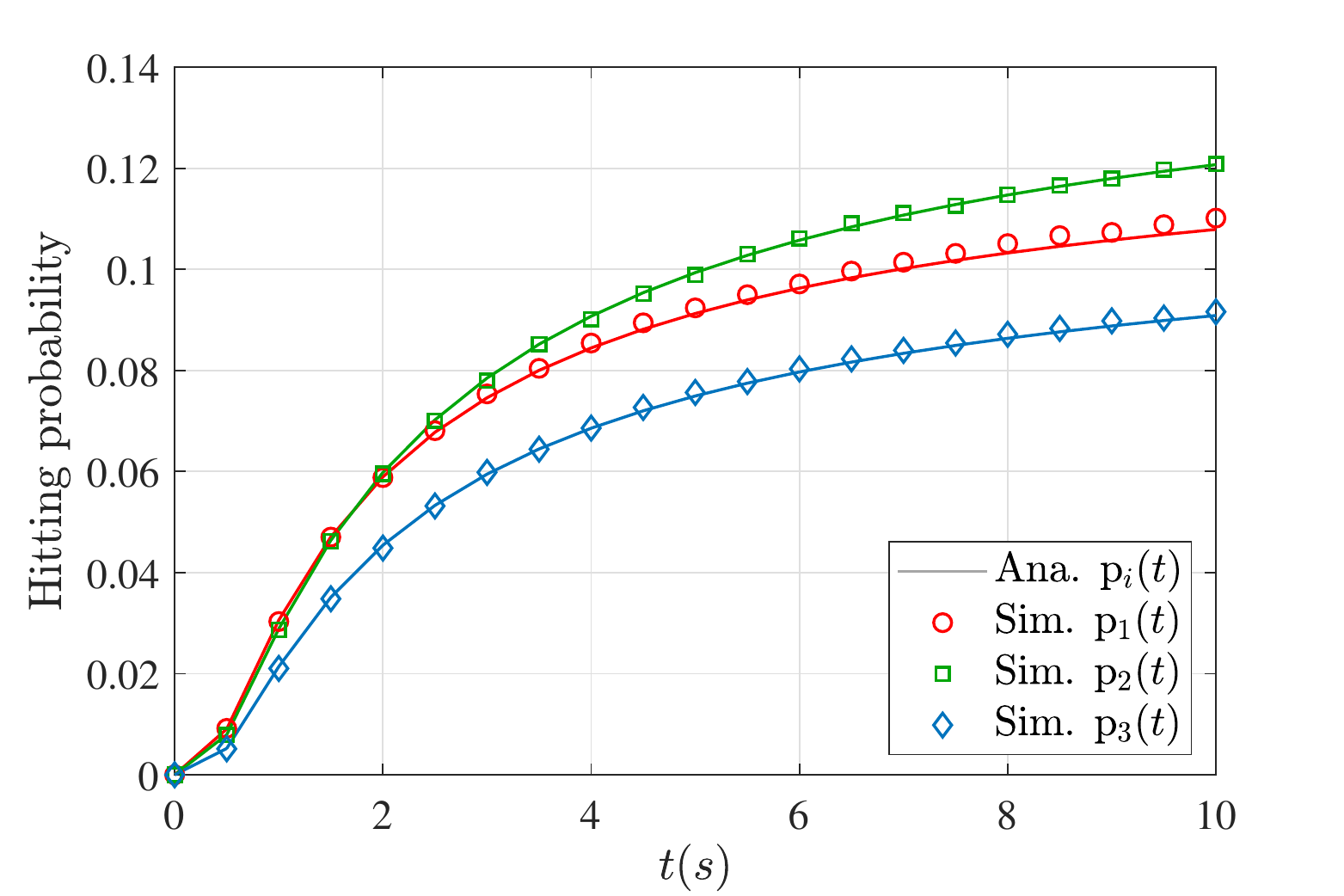}
	\caption{Variation of $\hp{i}(t)$ with time $t$ for a MC system with three FARs. Parameters: $\x_1= [25,0,0],\ \x_2=[-25,5,0],\ \x_3=[20,-15,10],\ a=5\ \mu m,\ D=100\mu m^2/s,\ \Delta t=10^{-5}\ s $.}
	\label{fig:3rxf1}
	\vspace{-0.5cm}
\end{figure}

\subsection{Accuracy of the hitting probability equations}\label{acc}

Fig.  \ref{fig:aberr} shows the variation of the absolute error in the derived hitting probability expression of FAR$ _1 $ 
defined as
\begin{align*}
\mathrm{ Absolute ~error} = \left| \mathrm{Analytical ~result} -\mathrm{ Simulation ~result}\right| , 
\end{align*}
when it is moved in the $y,z$-plane, while FAR$ _2 $ and FAR$ _3 $ are positioned at fixed locations. The shaded regions in gray represents the excluded locations as  the FARs overlap with each other in this region. Fig. \ref{fig:aberr} confirms that the derived equation has good accuracy in most of the regions. However, similar to \cite{sabu2020a}, the accuracy of the derived equation can be bad if the FARs in the medium are too close to each other and the FAR of interest hinders the other FARs from the line of sight of the point transmitter.\par

\begin{figure}[ht]
	\centering
	\includegraphics[width=.65\linewidth]{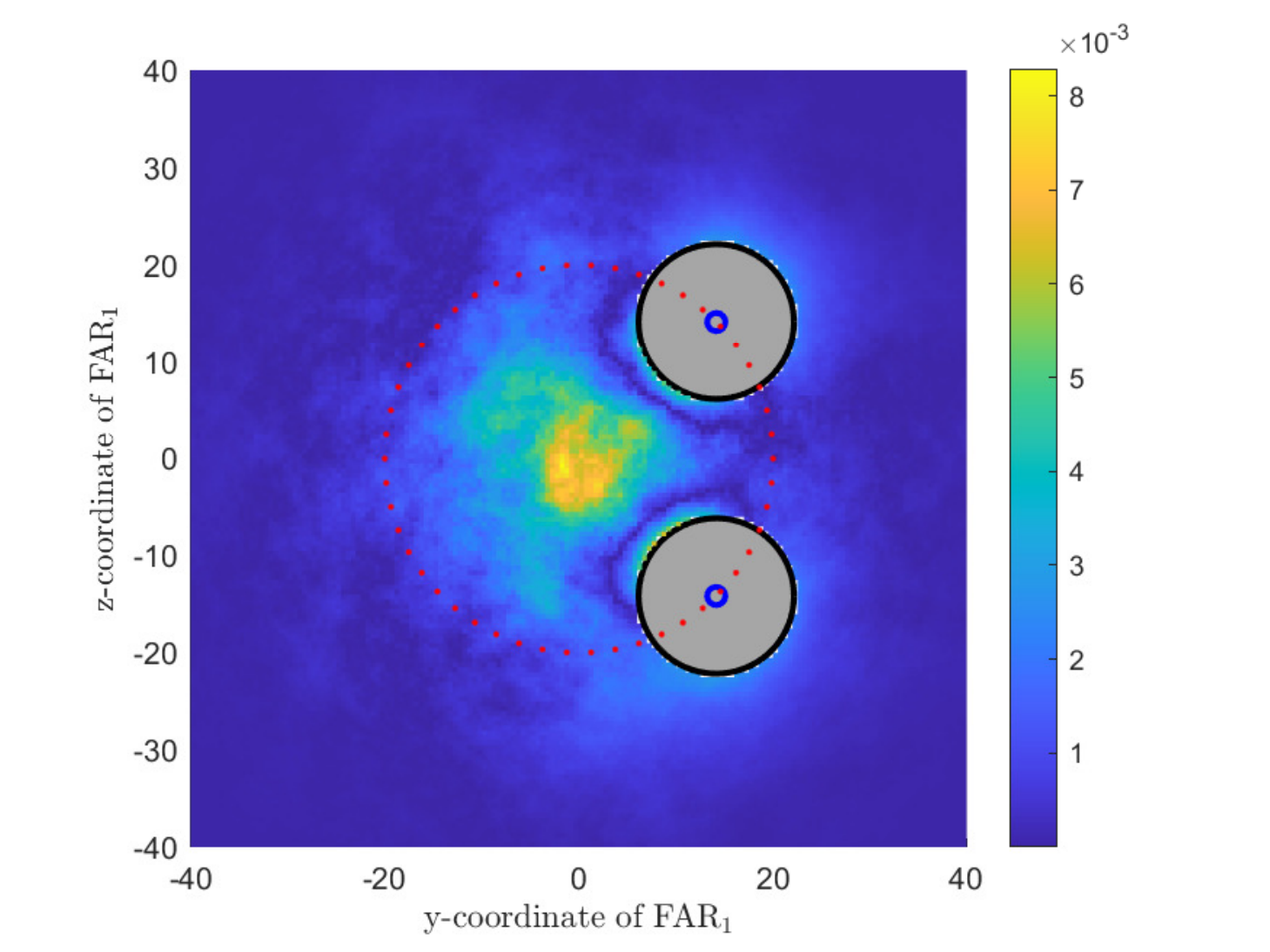}
	\caption{Variation of the absolute error of the hitting probability of FAR$ _3 $ in a $ 3- $FAR system. Parameters: $ a=5\mu m,\ D=100\mu m^2/s,\ t=1s,\ \Delta t=10^{-4}s,\ \x_2=[10,\	14.14,\	14.14],\  \x_3=[10,\	14.14,\	-14.14],\  \x_1=[10,\ y,\	z], \text{where } y,z$ is varied. 
	}
	\label{fig:aberr}
	\vspace{-0.4cm}
\end{figure}

The analytical expression given in Theorem \ref{t3rx} simplifies when the FARs are equidistant from the transmitter and equidistant from each other, which we consider next. 

\subsection{Special case: Symmetric system with equidistant FARs}\label{secuca}
\begin{figure}[ht]
	\centering
	\includegraphics[width=.65\linewidth]{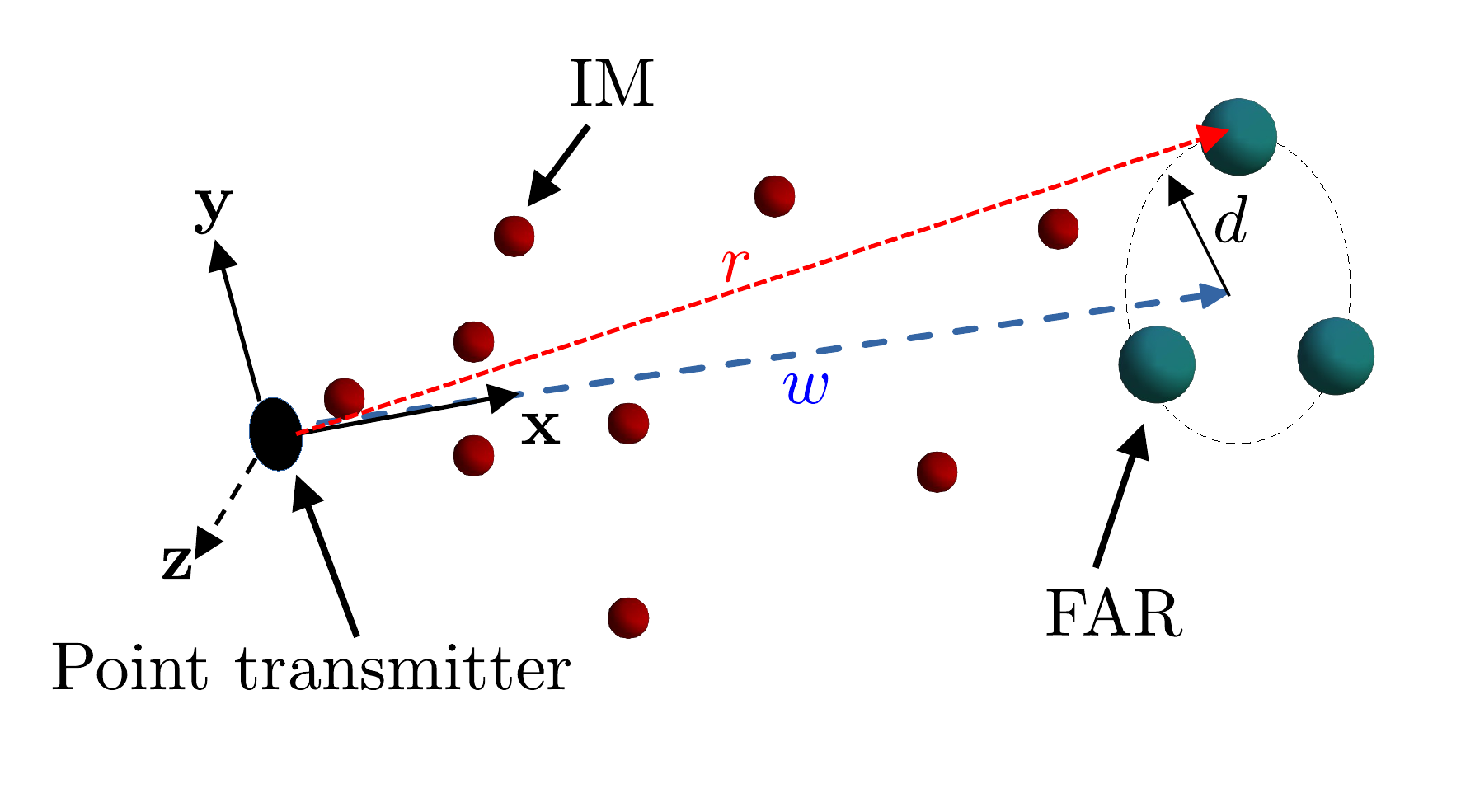}
	\caption{An MC system with three FARs equidistant from each other and the transmitter.}
	\label{fig:uca1}
\end{figure}

Now, we consider a symmetric system where all three FARs are equidistant  the transmitter and their mutual distance is also the same as shown in Fig. \ref{fig:uca1}. The point transmitter is at the origin and the centers of the FARs are located in a circle with center $ [w,0,0] $ and radius $ d $. Therefore the location of the FAR$ _i $ is $\x_i= [w,d\cos\left(2\pi (i-1)/3\right),d\sin\left(2\pi (i-1)/3\right)], \ i\in\{1,2,3\} $. This type of arrangement of FARs is called uniform circular array (UCA) \cite{gursoy2019,tang2021}. 
The hitting probability for this special case 
is given in the following theorem.

\begin{theorem}\label{c3rxeq}
	
	In an MC system with three FARs that are equidistant from the transmitter at the origin and equidistant from each other (\ie $ r_i=r \text { and } R_{ij}=R,\ i,j\in{1,2,3},\ i\neq j$), the probability that an IM emitted from the point source at the origin  hits  $\mathrm{FAR} _1 $ within time $ t $ in the presence of other two FARs is 
	\begin{align}
		\hp{1}(t)=\frac{a}{r}\sum_{n=0}^{\infty}\frac{(-2a)^n}{R^n}\erfc\left(\frac{r-a+n(R-a)}{\sqrt{4Dt}}\right).\label{eq3rxeq}
	\end{align}
\end{theorem}
\begin{proof}
	See Appendix \ref{a3rxeq}.
\end{proof}
\begin{coro}\label{corolim}
	The fractions of IMs that hits each of the FARs eventually (\ie $ t\rightarrow\infty $) in an MC system with three FARs that are equidistant from the transmitter at the origin and equidistant from each other is given by
	\begin{align}
		\hp{\infty}=\frac{a/r}{1+2a/R}.
	\end{align}
\begin{proof}
	Applying the final value theorem for the Laplace transform (\ie $ \lim_{t\rightarrow\infty} f(t)=\lim_{s\rightarrow0}sF(s)$) \cite{spiegel1992} in \eqref{eq3rxeq} gives Corollary \ref{corolim}.
\end{proof}
\end{coro}


Fig. \ref{fig:htd} shows the variation of hitting probability (\ie eq. \eqref{eq3rxeq}) to the diffusion coefficient for different values of FAR radius. The equation \eqref{eq3rxeq} is validated using particle-based simulations. As $ D $ increases, the hitting probability increases due to the faster motion of the IMs. The increase in the hitting probability with the increase in $ a $ is due to the increase in the surface area of the FARs and thereby the increase in the probability of capturing the IM. 

  \begin{figure}[ht]
	\centering
	\includegraphics[width=0.65\linewidth]{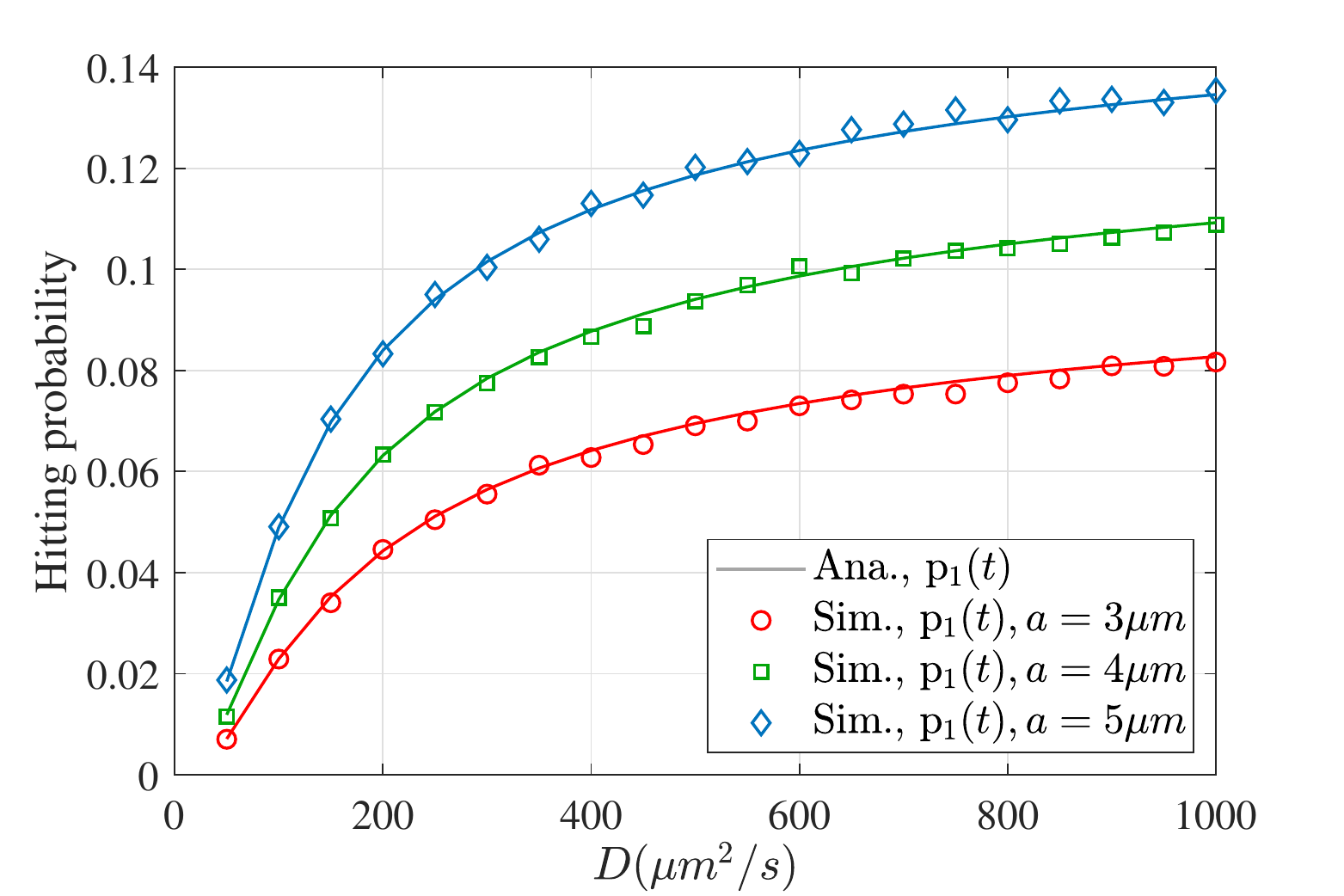}
	\caption{Variation of the hitting probability $\hp{1}(t)$ with diffusion coefficient $ D $ for different values of $ a $ for the symmetric case. Here, $d=20 \mu m, w=10\mu m,\ r=22.36\mu m,\ 
		t=10^{-4}\ s,\ t=1s $.}
	\label{fig:htd}
	\vspace{-0.3cm}
\end{figure}

We now discuss two specific scenarios to demonstrate the applicability of the desired result in the next two subsections.

\subsection{Impact of malicious receivers}
Let us consider a scenario where FAR$_1$ is the intended receiver for the transmitter, and the rest of the two FARs are acting as malicious receivers. 
The relative change in the IMs that were supposed to hit the intended FAR$ _1 $ within time $ t $, but are hitting the other two FARs  first is given by
	\begin{align}
		\hpdiff{1}(t)=\frac{\hpone(t,r_1)-\hp{1}(t)}{\hpone(t,r_1)}.\label{ediff}
	\end{align}
Note that, $ \hpdiff{1}(t) $ quantifies the relative influence of the malicious receivers on the desired receiver in capturing the IMs. By useful IMs, we mean the IMs that would reach FAR$ _1 $ and are responsible for detection. The more the $ \hpdiff{1}(t) $ is, the more the IMs captured by the malicious FARs would be. This will result in the degradation of hitting probability of the intended receiver, degrading its communication performance. To avoid the malicious FARs to degrade the desired communication, $ \hpdiff{1}(t) $ has to be minimal.

\begin{coro}
The relative change in the IMs that are supposed to hit the intended FAR eventually, but are hitting the other two FARs in the symmetric system (\ie all the FARs that are equidistant from the transmitter at the origin and equidistant from each other) is given 
	by\begin{align}
		\hpdiff{\infty}&=
		\frac{1}{1+R/2a}.
	\end{align}
\end{coro}

Therefore, the malicious receivers' influence is minimal when $a$ is small and $ r $ and $ R $ are large. For a secure MC system, the eavesdropping by the malicious FARs is minimal when the FARs are of small radii and located far from the transmitters and each other. This can also be verified from Fig. \ref{fig:qfig} which shows the variation of this influence with respect to the angular separation between the receivers. With sufficient angular separation, effect of malicious receivers can be reduced signicantly.

  \begin{figure}[ht]
	\centering
	\includegraphics[width=0.65\linewidth]{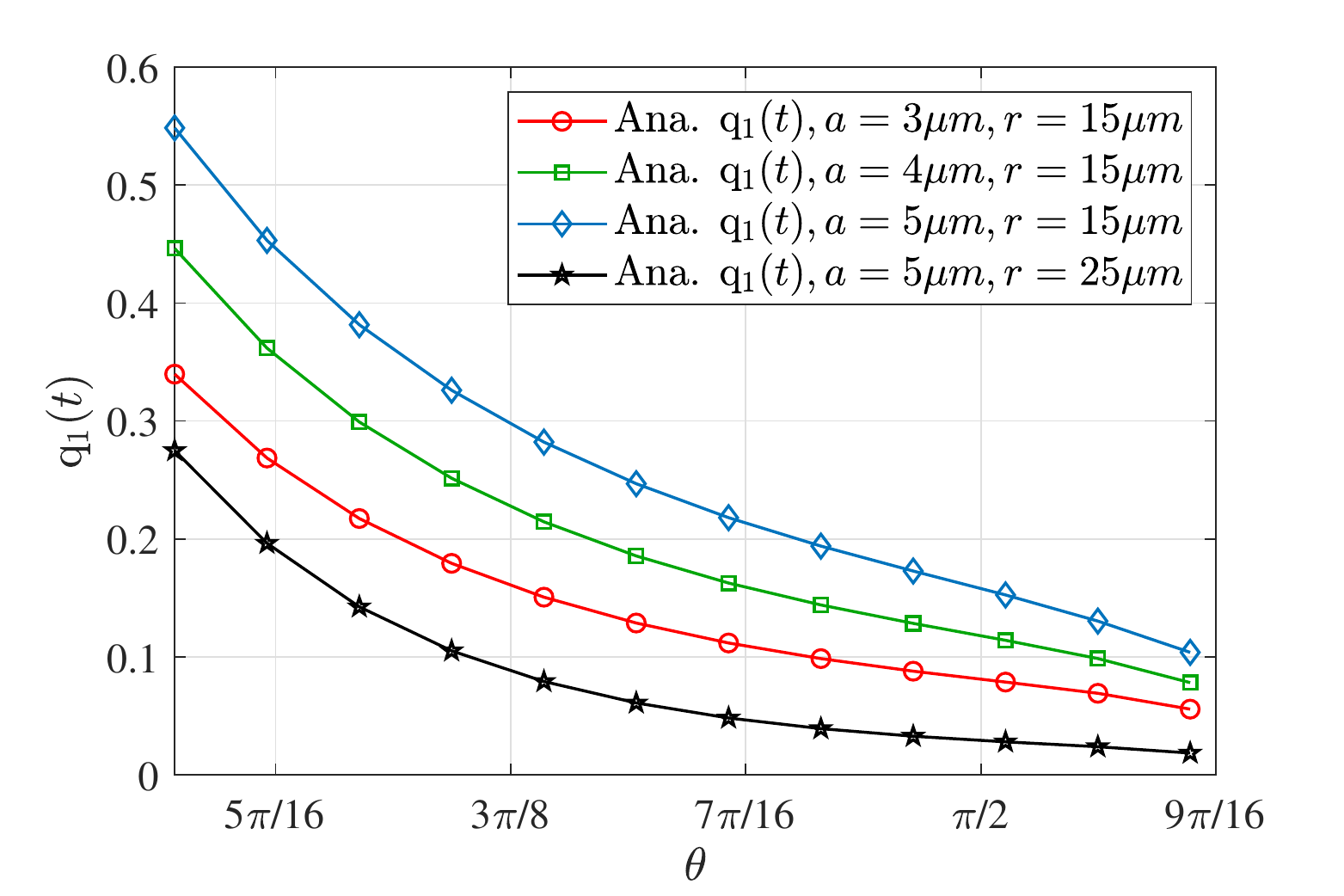}
	\caption{Variation of $  \hpdiff{1}(t) $ with angle between the intended and malicious receivers (\ie $ \phi_{12}=\phi_{13}=\theta $). Here, the intended receiver is located at $ \x_1=[r,0,0] $ and the malicious receivers are located $ \x_2=[r\cos(\theta),r\sin(\theta),0] $ and $ \x_3=[r\cos(-\theta),r\sin(-\theta),0] $ with  $D=100 \mu m^2/s$.}
	\label{fig:qfig}.
	\vspace{-0.5cm}
\end{figure}

  \begin{figure}
	\centering
	\includegraphics[width=0.65\linewidth]{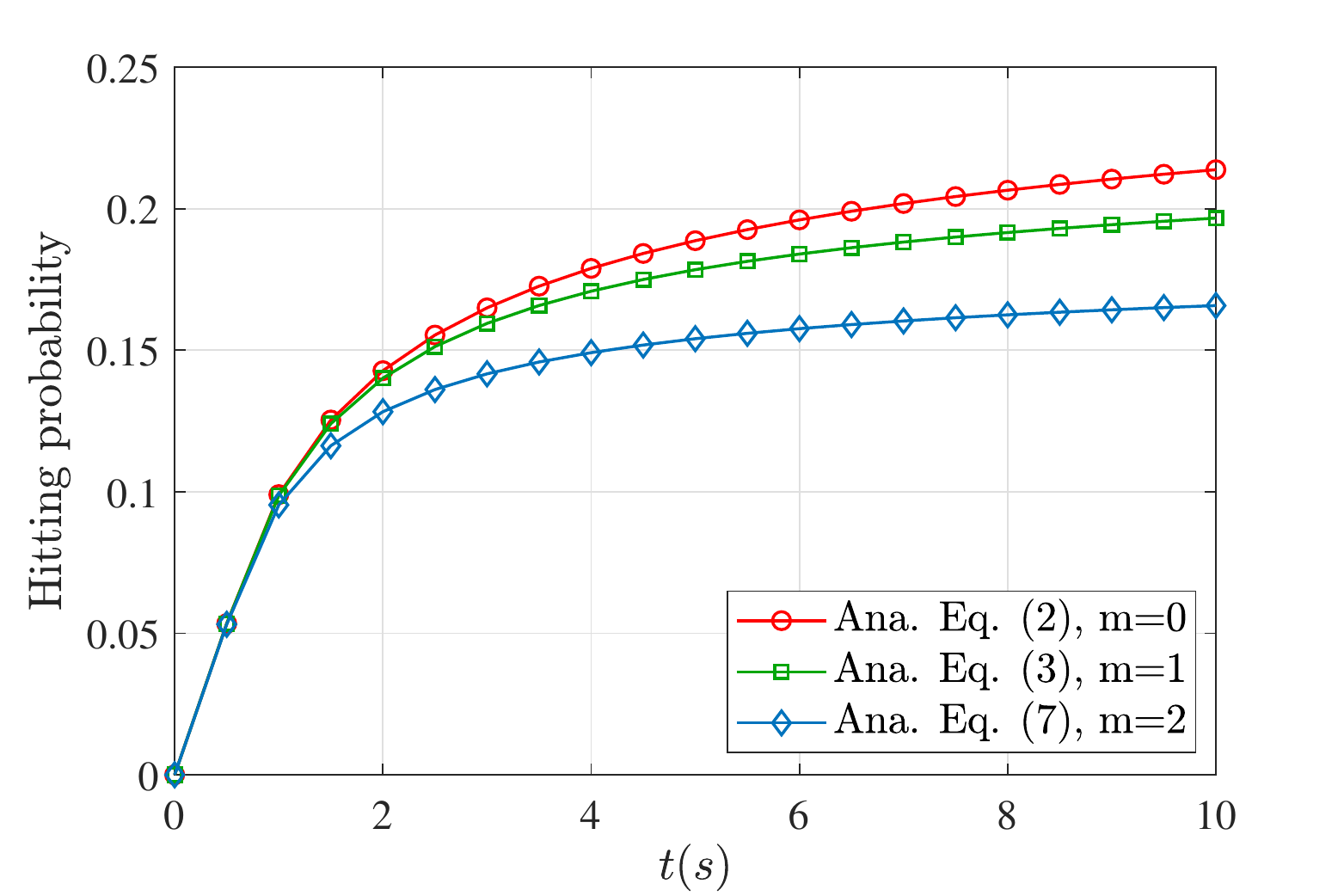}
	\caption{Variation of the hitting probability $\hp{1}(t)$ with $ t $ in the presence of $m$ malicious receivers. Here, receivers are located according to Fig. \ref{fig:uca1} with $d=20 \mu m, w=10\mu m,\ r=22.36\mu m,\ a=4\mu m,\
		t=10^{-4}\ s,\ t=1s $.}
	\label{fig:nuca}
	\vspace{-0.5cm}
\end{figure}
Fig. \ref{fig:nuca} shows the variation of the hitting probability within time $ t $ with $ t $ for $ m$ number of malicious receivers with values $m=0,1,2$. Here, $m=0$ refers to the case when no malicious receiver is present. When the number of malicious receivers is increased, the hitting probability decreases due to the increase in their influence. The widening gap between the curves shows that the degrading effect of these receivers grows with time. Therefore, it is essential to consider the effect of this influence while designing MC systems with multiple malicious FARs.

\subsection{Impact of cooperating receivers}
Let us consider a scenario where FAR$_1$ is the intended receiver, and the rest of the two FARs cooperate in the receiving. In this case, the receiver can detect the transmission if any of the FARs have received IMs. The fractions of IMs  hitting any of the FARs within time $t$ to the fraction of IMs hitting a single FAR in the absence of other FARs  is given by
	\begin{align}
		\mathsf{s}_1(t)&=
		\frac{\sum_i\hp{i}(t)}
		{\hpone(\infty,r_1)}.
	\end{align}
$ \mathsf{s}_1(t) $ can be considered as the array gain when multiple FARs are employed instead of a single FAR. 
	
\begin{coro}
	The fractions of IMs eventually hitting any of the FARs  to the fraction of IMs hitting a single FAR in the absence of other FARs  in a symmetric system (\ie all the FARs that are equidistant from the transmitter at the origin and equidistant from each other) is given by
	\begin{align}
		\mathsf{s}_\infty&=\frac{3\hp{\infty}}{\hpone(\infty,r)}
		=\frac{3}{1+2a/R}.
	\end{align}
\end{coro}

Note that the asymptotic array gain $\mathsf{s}_\infty$ is less than 3. Further, it decreases when FARs are placed closer to each other (\ie $ R $ decreases). When FARs are far apart (\ie $ R\rightarrow\infty $), $ \mathsf{s}\rightarrow 3 $.

Having derived the hitting probability for $N=3$ FARs, we  will now show how we can extend the proposed analysis for the general value of $N$ FARs.

\section{Channel Model for a  MC system with $N$ FARs 
}
We now consider that there are $ N $ FARs coexisting in the same communication channel. 
Let us first focus on FAR$_i$. 
We assume that the point transmitter at the origin emits an IM at time $ t=0 $. The probability that this IM  first hits the surface of FAR$_{j}$ 
($j\neq i $) in the time interval $[\tau,\tau+\dd\tau] $ is  $\displaystyle \frac{\partial \hp{j}(\tau)}{\partial \tau} \dd\tau$. If the IM were not absorbed by FAR$ _j $, it would have eventually hit the FAR$ _i $ in the remaining $ t-\tau $ time with probability $\hpone(t-\tau,R_{ji})$, where $R_{ji}$ is the distance between the nearest point ($ \y_j $) on the surface of FAR$_{j} $  from the transmitter and the center of FAR$_{i} $ ($ \x_i $). Note that, the hitting point of IM at FAR$_j$ is approximated  by $\y_j$.

Similar to Appendix \ref{a3rx}, the hitting probability of an IM that was supposed to hit FAR$ _{i} $ within time $ t $ but is hitting the other $ N-1 $ FARs is given by 
\begin{align}
	\hpone(t,r_i) - \hp{i}(t)   =&  \sum_{j=1,\ j\neq i}^{N} 
	\int_{0}^{t} \frac{\partial \hp{j}(\tau)}{\partial \tau}\hpone(t-\tau,R_{ji}) 	\dd\tau.\label{eqtnxr} 
\end{align}
Now taking the Laplace transform of \eqref{eqtnxr} gives
\begin{align}
	\lhpone(s,r_i) -\lhp{i}(s)  =  \sum_{j=1,\ j\neq i}^{N} 
	s\lhp{j}(s) \lhpone(s,R_{ji}) , \ \ \forall\ i.\label{eqsnxr} 
\end{align}
Note that, \eqref{eqsnxr} forms a system of $ N $ linear equations with $ N $ unknowns for $ i\in\{1,2,\cdots,N\} $. Solving the system of linear equations gives the value of hitting probability at the $ i $th FAR, \ie $ \hp{i}(t) $. The solution to the equations and further analysis is left for the future work.
\vspace{-0.2cm}
\section{Conclusions}
Multiple absorbing receivers may coexist in the same communication channel in a molecular communication system. It is essential to obtain the analytical channel model for analyzing and understanding such systems. However, exact channel models are not available in the literature due to the mathematical intractability in solving the diffusion equations. In this work, we bridged this gap by deriving an approximate analytical expression for the hitting probability of an IM emitted by a point transmitter of each FARs. The derived hitting probability expression was simplified by considering that the FARs are equidistant from each other and the transmitter. Using the derived expressions, we discussed several design insights like quantifying the mutual influence of FARs, the fraction of IMs eventually hitting the FARs, and the fractions of IMs eventually hitting all of the FARs to the fraction of IMs hitting a single FAR.  We also provided an analytical framework to derive the approximate hitting probability expression for an MC system with $ N $ FARs, and further analysis will be done as future work. The derived analytical expressions can be used to analyze and design SIMO and MIMO systems to improve the reliability and data rate. It also helps design communication systems that secure the desired receiver bio-nanomachines from eavesdroppers.
\vspace{-0.2cm}
\appendices
	\section{Derivation of Fraction of IMs Absorbed at Each FAR in a $3-$FARs System}\label{a3rx}

Consider that the point transmitter at the origin emits an IM at time $ t=0 $. The probability that this IM  first hits the surface of FAR$_{2} $ in the time interval $[\tau,\tau+\dd\tau] $ is  $\frac{\partial \hp{2}(\tau)}{\partial \tau} \dd\tau$. If not absorbed by FAR$ _2 $, the IM would have eventually hit the FAR$ _1 $ in the remaining $ t-\tau $ time with probability $\hpone(t-\tau,R_{21})$, where $R_{21}$ is the distance between the nearest point ($ \y_2 $) on the surface of FAR$_{2} $  from the transmitter and the center of FAR$_{1} $ ($ \x_1 $). Here, the hitting point of IM at FAR$_2$ is approximated  by $\y_2$. Similarly, the IM emitted from the point source at the origin first hits the surface of FAR$ _3 $ in the time interval $[\tau,\tau+\dd\tau] $ with probability  $\frac{\partial \hp{3}(\tau)}{\partial \tau} \dd\tau$. If not absorbed by FAR$_3 $, the IM would have eventually hit the FAR$ _1 $ in the remaining $ t-\tau $ time  with probability $\hpone(t-\tau,R_{31})$, where $R_{31}$ is the distance between the nearest point ($ \y_3 $) on the surface of FAR$_{3} $  from the transmitter and the center of FAR$_{1} $ ($ \x_1 $). Therefore, 
the hitting probability of an IM that was supposed to hit FAR$ _{1} $ within time $ t $ but is hitting either FAR$ _{2} $ or FAR$ _{3} $ is given by 
\begin{align}
	\hpone(t,r_{1}) - \hp{1}(t)   =&   
	\int_{0}^{t} \frac{\partial \hp{2}(\tau)}{\partial \tau}\hpone(t-\tau,R_{21}) 	\dd\tau  
	+ 
	\int_{0}^{t} \frac{\partial \hp{3}(\tau)}{\partial \tau}\hpone(t-\tau,R_{31}) \dd\tau , \label{eqn1}
\end{align}
Similarly, the hitting probability of an IM that was supposed to hit FAR$ _{2} $ within time $ t $ but is hitting either FAR$ _{1} $ or FAR$ _{3} $, and the hitting probability of an IM that is supposed to hit FAR$ _{3} $ within time $ t $ but is hitting either FAR$ _{1} $ or FAR$ _{2} $  first is given by 
\begin{align}
	\hpone(t,r_{2}) - \hp{2}(t)   = &  
	\int_{0}^{t} \frac{\partial \hp{3}(\tau)}{\partial \tau}\hpone(t-\tau,R_{32}) \dd\tau  
	+ \int_{0}^{t} \frac{\partial \hp{1}(\tau)}{\partial \tau}\hpone(t-\tau,R_{12}) \dd\tau ,
	\label{eqn2}\\
\text{and }
	\hpone(t,r_{3}) -  \hp{3}(t)   = &   
	\int_{0}^{t} \frac{\partial \hp{1}(\tau)}{\partial \tau}\hpone(t-\tau,R_{13}) \dd\tau  
	+ \int_{0}^{t} \frac{\partial \hp{2}(\tau)}{\partial \tau}\hpone(t-\tau,R_{23}) \dd\tau 
	\label{eqn3}  ,
\end{align}
respectively. Taking Laplace transform on both sides of \eqref{eqn1}, \eqref{eqn2} and \eqref{eqn3} gives the following set of equations: 
\begin{align}
	\lhpone(s,r_1) -\lhp{1}(s)  =  
	s\lhp{2}(s) \lhpone(s,R_{21}) +
	s\lhp{3}(s)\lhpone(s,R_{31}),\label{eqn4}\\
	\lhpone(s,r_2) -\lhp{2}(s)   =  
	s\lhp{1}(s) \lhpone(s,R_{12}) 
	+ s \lhp{3}(s) \lhpone(s,R_{32}),\label{eqn5}\\
	\lhpone(s,r_3) - \lhp{3}(s)   =  
	s\lhp{1}(s)\lhpone(s,R_{13})
	+ s\lhp{2}(s) \lhpone(s,R_{23}).	\label{eqn6}
\end{align}
Here $\lhp{i}(s)$ and $\lhpone(s,x),\ x\in\{r_i,R_{ij}\}\ i,j\in\{1,2,3\}, i\neq j
$ are the Laplace Transforms of $\hp{i}(t)$ and $\hpone(t,x)$ respectively. Note that,
$ \lhpone(s,x)=\frac{a}{sx}\exp\left(-\left(x-a\right)\sqrt{\frac{s}{D}}\right).\label{lapeq} $

Solving \eqref{eqn4}, \eqref{eqn5} and \eqref{eqn6} for the value of $\lhp{1}(s)$ and taking the inverse Laplace transform gives \eqref{eq3rx}.
%
\section{Derivation of Theorem \ref{c3rxeq}}\label{a3rxeq}
 When three FARs that are equidistant from the transmitter at the origin and equidistant from each other, $ r_i=r,\ \forall i  $ and $ R_{ij}=R,\ \forall i,j,\  i\neq j $. Therefore \eqref{lp3rx} simplifies to
\begin{align}
	\lhp{1}(s)   =
	\frac{\lhpone(s,r) -s\mathsf{a}(s)+s^2\mathsf{b}(s)}
	{1-s^2\mathsf{c}(s)+s^3\mathsf{d}(s)}, \label{abeq19}
\end{align}\\
where $\mathsf{a}(s){=}2\lhpone(s,r)\lhpone(s,R) ,\ \mathsf{b}(s){=}\lhpone(s,r)\lhpone^2(s,R),\ \mathsf{c}(s)=3\lhpone^2(s,R),\  \mathsf{d}(s){=}2\lhpone^3(s,R) $.
Simplifying \eqref{abeq19} further gives
\begin{align}
	\lhp{1}(s)  &
	=\lhpone(s,r)\frac{1}{1+2s\lhpone(s,R)}\label{l3rx}\\
	&\stackrel{(a)}{=}\lhpone(s,r)\times\sum_{n=0}^{\infty}(-2)^n\left(s\lhpone(s,R)\right)^n,
	\label{abeq20}
\end{align}
where  $ (a) $ is due to the identity $ (1+2x)^{-1}=\sum_{n=0}^{\infty}(-2)^nx^n $.
Substituting \eqref{lapeq} in \eqref{abeq20} gives
{\small\begin{align}
	\lhp{1}(s)=\frac{a}{r}\sum_{n=0}^{\infty}\frac{(-2a)^n}{sR^n}\exp\left(-(r-a+n(R-a)\sqrt{\frac{s}{D}})\right).\label{eqcor}
\end{align}}
Applying the inverse Laplace transform on both sides of \eqref{eqcor} gives Theorem \ref{c3rxeq}.

	\bibliographystyle{IEEEtran}
	\bibliography{nrx_ref}
\end{document}